\pdfoutput=1
\documentclass[amsart]{revtex4-2}

\usepackage{beitmacros}
\usepackage{hyperref, hypernat}
\usepackage{enumitem}
\usepackage{multirow}
\usepackage{fixme}

\newcommand{\dg}{\dagger}

\newcommand{\nuc}[1]{||{#1}||_{\textrm{nuc}}}
\newcommand{\Err}{\textrm{Err}}
\newcommand{\Total}{\textrm{Total}}
\newtheorem{prop}{Proposition}

\newcommand{\beit}{\affiliation{BEIT sp. z o o., Mogilska 43, 31-545 Kraków, Poland}}%
\newcommand{\beitemailkd}{\email[]{konrad@beit.tech}}%
\newcommand{\beitemailez}{\email[]{emil@beit.tech}}%
\newcommand{\beitweb}{\homepage{https://www.beit.tech}}%

\begin{document}
\title{Simultaneously optimizing symmetry shifts and tensor factorizations 
for cost-efficient Fault-Tolerant Quantum Simulations of electronic Hamiltonians}
\author{Konrad Deka and Emil Zak}\beit\beitemailkd\beitemailez\beitweb

\date{\today}

\begin{abstract}
In fault-tolerant quantum computing, the cost of calculating Hamiltonian eigenvalues using the quantum phase estimation algorithm is proportional to the constant scaling the Hamiltonian matrix block-encoded in a unitary circuit. We present a method to reduce this scaling constant for the electronic Hamiltonians represented as a linear combination of unitaries.
Our approach combines the double tensor-factorization method of Burg et al.~\cite{burg} with the the block-invariant symmetry shift method of Loaiza and Izmaylov~\cite{izmaylov1}. By extending the electronic Hamiltonian with appropriately parametrized symmetry operators and optimizing the tensor-factorization parameters, our method achieves a 25\% reduction in the block-encoding scaling constant compared to previous techniques. 
The resulting savings in the number of non-Clifford T-gates, which are an essential resource for fault-tolerant quantum computation, are expected to accelerate the feasiblity of practical Hamiltonian simulations. We demonstrate the effectiveness of our technique on Hamiltonians of industrial and biological relevance, including the nitrogenase cofactor (FeMoCo) and cytochrome P450.
\end{abstract}

\maketitle

\section{Introduction} 
One of the most promising applications of quantum computing is the simulation of quantum chemistry Hamiltonians, a task that exceeds the capabilities of classical computing for certain classes of industrially relevant molecules. A central focus in this field is the computation of Hamiltonian eigenvalues, which have been extensively studied theoretically~\cite{burg, thc, Low2019, Berry2019}. The eigenvalues of electronic Hamiltonians correspond to the energy levels of electronic states. Differences in these energy levels for varying molecular structures and compositions inform about the relative stability of compounds and pathways for chemical transformations. Calculating reaction pathway diagrams is pivotal for designing innovative materials, including solvents, donor/acceptor molecules for energy storage~\cite{li-ionpaper}, catalysts for nitrogen fixation~\cite{Beinert1997, Reiher} and for artificial photosynthesis~\cite{losalamos24}.

In this work, we focus on two molecules for which quantum chemistry calculations at the \textit{chemical accuracy} of 1 mHartree or better are currently beyond the reach of classical computers~\cite{Goings2022,nike}. The first is the FeMo cofactor (Fe$_7$MoS$_9$C) of the nitrogenase enzyme, a catalyst participating in the biological nitrogen fixation process~\cite{Reiher}. The second system is the Cpd I structure in the cytochrome P450 enzyme~\cite{Shaik2009}, a protein that governs the metabolism of drugs in humans~\cite{Liu2007,Boobis2009}. Structurally similar to the heme molecule in the hemoglobin protein~\cite{Cho2008}, Cpd I represents a class of biologically significant molecules. 

Density Matrix Renormalization Group (DMRG) studies of FeMoCo~\cite{nike} and Cpd I~\cite{Goings2022} suggest that quantum computation can be advantageous for these systems. The largest DMRG calculations to date for FeMoCo~\cite{nike}, in the spin $S=3/2$ ground state with 76 orbitals and 113 electrons (forming a Hilbert space of approximately $10^{35}$ elements), highlight the need for carefully choosing the basis set to adequately capture electron correlation effects, particularly for open-shell molecules. For this reason, achieving chemical accuracy for these compounds with reliable error estimates remains challenging for classical computing methods~\cite{Marti2010,nike,Elfving2020HowWQ}.

Quantum computational methods, in principle, enable removing limitations of classical algorithms for systems with high complexity in three key aspects. One limitation is memory, which quantum computation addresses by using qubits. Another is the ability to calculate eigenvalues of a given input Hamiltonian with arbitrary precision and accuracy, directly controlled by the simulation parameters. Like classical computation, quantum computation utilizes a basis set representation of the electronic Hamiltonian, which can be systematically extended and improved, to increase the calculation accuracy, as ensured by the variational principle~\cite{Helgaker2000}. Finally, the total computation time can be shorter for a quantum computer implementation~\cite{Elfving2020HowWQ}. 

Quantifying the advantage carried by a quantum computation over classical methods for electronic structure calculations requires a method for quantum resource estimation. On the classical computing end, it is necessary to determine the accuracy of the energy levels and the associated computational complexity for a given Hamiltonian. The verdict of quantum advantage depends on these metrics and is a subtle, actively debated topic~\cite{Lee2023,Daley2022,krol2024}. Quantum resource estimation is fundamental to assessing quantum advantage as it involves counting the number of qubits and non-Clifford T-gates, with the latter setting the runtime for the quantum computation~\cite{gidney2024}. Reducing this runtime is essential for advancing the practical utility of a quantum computation. In this work, we focus on minimizing this algorithmic runtime to accelerate quantum electronic structure computation and to reduce their overall cost. 

For our quantum computation framework we choose quantum phase estimation (QPE)~\cite{Abrams1999}, being the most efficient algorithm for calculating Hamiltonian eigenvalues on a fault-tolerant quantum computer~\cite{thc}. In QPE, the phases generated by time-evolution under Hamiltonian $H$ are accumulated by appropriate input eigenstates of the Hamiltonian at a controllable precision $\epsilon$. The time evolution operator, $e^{iHt}$, a core component of QPE, is most efficiently implemented for electronic Hamiltonians using qubitization~\cite{Low2019,Berry2019,burg}. To perform qubitization, one needs to construct a \emph{block encoding} of the Hamiltonian $H$, that is, a quantum circuit $C$ whose unitary matrix contains $H / \lambda$ (where $\lambda > 0$) as a submatrix. The implementation of $e^{iHt}$ then requires $\mathcal{O}(\lambda t)$ repetitions of circuit $C$.
Since the runtime of the QPE algorithm scales linearly with the block-encoding scaling constant $\lambda$, lowering its value has been subject to several recent studies~\cite{burg,thc,izmaylov1,cohn,oumarou, ollitrault}. 
In this work we consider the electronic Hamiltonian studied also in Burg et al.~\cite{burg}, which can be written as:
\begin{equation} \label{eqn:hamiltonian_form_first}
	H = \sum_{i, j, \sigma} h_{ij} a_{i \sigma}^\dg a_{j \sigma}
	+ \sum_{i, j, \sigma, k, l, \tau } g_{ijkl} a_{i \sigma}^\dg a_{j \sigma} a_{k \tau}^\dg a_{l \tau},
\end{equation}
where $0 \le i, j, k, l < N$ are the orbital indices, $\sigma$ and $\tau$ denote spin indices. $a_{i\sigma}$ annihilates electron in spin-orbital $(i,\sigma)$. $h_{ij}$ represents the one-body terms, and $g_{ijkl}$ describes two-body electron-electron interactions. Matrix elements of the electronic Hamiltonians for FeMoCo and Cpd I were reported in previous works~\cite{Reiher,burg,rocca,Goings2022}, which also discuss the associate error estimates. For demonstrating our technique, here we adapt the electronic Hamiltonians from refs~\cite{Reiher,Goings2022} for FeMoCo and Cpd I (cytochrome P450), respectively. 

The block-encoding scaling constant $\lambda$ can be lowered by representing the electron-electron interaction tensor $g_{ijkl}$ with lower-rank tensors. Burg et al.~\cite{burg}, introduced the \textit{double factorization} algorithm (DF), for the electronic ground state energy calculations using a Hamiltonian represented with lower-rank tensors. 
DF involves performing eigenvalue decomposition of the $g_{ijkl}$ tensor twice, giving a block-encoding scaling constant we denote as $\lambda_{DF}$. Since then, further improvements in the value of $\lambda$ have been made by choosing alternative decomposition schemes, including \emph{tensor hypercontraction} method of Lee et al.~\cite{thc}, \emph{compressed double factorization} (CDF) by Cohn et al.~\cite{cohn} and \emph{regularized compressed double factorization} (RCDF) by Oumarou et al.~\cite{oumarou}.

A qualitative progress in lowering $\lambda$ came through with a \textit{block-invariant symmetry shifts} (BLISS) method by Izmaylov and Loaiza~\cite{izmaylov1}. BLISS assumes the Hamiltonian $H$ is applied to a state $\ket{\psi}$, meaning that it is sufficient to consider a different Hamiltonian  $\tilde{H}$,
as long as $(H - \tilde{H}) \ket{\psi} = 0$. In the case of electronic Hamiltonians, the state $\ket{\psi}$ that enters the QPE algorithm describes a fixed number of electrons, i.e $(N_e - n_e) \ket{\psi} = 0$ where $n_e$ is a positive integer and $N_e := \sum_{i \sigma} a_{i \sigma}^\dg a_{i \sigma}$ is the total electron number operator.
In ref.~\cite{izmaylov1} the symmetry-shifted Hamiltonian is chosen as $\tilde{H} := H + B(N_e - n_e)$, where 
$B$ stands for a Hermitian operator. Izmaylov and Loaiza apply this observation to the double-factorization technique as well as other factorizations~\cite{izmaylov1, izmaylov2}, demonstrating a numerical reduction in the value of $\lambda$ for all factorizations considered.
Recently, Rocca et al.~\cite{rocca} introduced the \emph{symmetry-compressed double factorization} (SCDF), improving on the CDF~\cite{cohn} and RCDF~~\cite{oumarou} methods in two ways. First, the circuit implementation of CDF and RCDF is more expensive than the original DF, which is addressed in~\cite{rocca}. Second, SCDF incorporates a simplified version of the method from ref.~\cite{izmaylov1} by replacing the original electronic Hamiltonian $H$ with $\tilde{H} := H + (b_1 N_e + b_2)(N_e - n_e)$, for optimized values of constants $b_1, b_2$.  

In this work, we propose a technique for further reducing $\lambda$ within the double factorization framework. 
We simplify the SCDF approach of Rocca et al.~\cite{rocca} and incorporate the full BLISS method~\cite{izmaylov1}. 
Our implementation achieves a minimum 25\% reduction in $\lambda_{DF}$ for complex systems such as FeMoCo~\cite{rocca,Reiher} and cytochrome P450 molecules~\cite{Goings2022}, compared with other techniques~\cite{rocca,izmaylov1,burg}.

\section{Methodology}
In this section, we lay out details of the double-factorization technique, formulated to naturally include the BLISS method, discussed in sec.~\ref{sec:symm-shift}.
\subsection{Definitions and conventions}
First, we introduce definitions and conventions used throughout the paper. 
Let $N \ge 1$ denote the number of orbitals encoded in the quantum states of $2N$ qubits, with each qubit representing a spin-orbital. The \emph{fermionic spin-orbital annihilation/creation operators}, $a_{j \sigma}, a^{\dag}_{j \sigma}$ respectively, destroy and create electrons in the spin-orbital $j,\sigma$, where $0 \le j < N$ and $\sigma \in \{0, 1\}$. These fermionic ladder operators satisfy the following canonical anti-commutation relations:
\begin{align}
	a_{j \sigma}^2 & = 0 \quad \textrm{ for all } j, \sigma, \\
	\left\{ a_{j \sigma}, a_{j \sigma}^\dg \right\} & = I \quad \textrm{ for all } j, \sigma,\\
	\left\{ a_{j \sigma}, a_{k \tau} \right \} & = 0 \quad \textrm{ if } j \neq k \textrm{ or } \sigma \neq \tau.
	\label{eq:commutation-relations-a}
\end{align}
The qubit encoding of these operators is achieved via the Jordan-Wigner transformation, which maps fermionic operators to Pauli operators as follows:
\begin{align}
a_{j \sigma} := \frac{1}{2} Z_0 Z_1 \dots Z_{j - 1 + N \sigma} \left( X_{j + N \sigma} + iY_{j + N \sigma} \right).\\
a_{j \sigma}^{\dg} := \frac{1}{2} Z_0 Z_1 \dots Z_{j - 1 + N \sigma} \left( X_{j + N \sigma} - iY_{j + N \sigma} \right).
\label{eq:annihilation-creation}
\end{align}
Here $Z_0$ is the Pauli-$Z$ gate acting on the first qubit (indexed from zero), and $iY_{j + N \sigma}$ represents the Pauli-$Y$ gate acting on the $(j+N\sigma)$-th qubit. We define the spin-independent, one-electron \textit{orbital excitation} operators by summing over the spin variable $\sigma$:
\begin{equation}
E_{ij} := \sum_{\sigma \in \{0, 1\} } a_{i \sigma}^{\dagger} a_{j \sigma}.
\label{eq:orbital-excitation}
\end{equation}
Using these definitions, the electronic Hamiltonian introduced in eq.~\ref{eqn:hamiltonian_form_first} can be expressed with the orbital excitation operators as: 
\begin{equation}
	H = \sum_{i,j = 0}^{N - 1} h_{ij} E_{ij} + 
	\sum_{i,j,k,l = 0}^{N - 1} g_{ijkl} E_{ij} E_{kl},
	 \label{eq:general_form_of_hamiltonian}
\end{equation}
where $h_{ij}, g_{ijkl}$ are real coefficients. The terms $h_{ij}$ represent the one-body matrix elements of the Hamiltonian (symmetric, $h_{ij} = h_{ji}$), while $g$ describe the two-body electron-electron interaction terms. The coefficients $g_{ijkl}$ exhibit the standard $8$-fold symmetry constraints: $g_{ijkl} = g_{jikl} = g_{ijlk} = g_{klij}$, as detailed in~\cite{Helgaker2000}.

\subsection{Block encoding of double-factorized electronic Hamiltonian}
\label{sec:block-encoding}
In this section, we discuss the double-factorization method for the electronic Hamiltonian and its block-encoding. For further details on the double-factorization approach, we refer the reader to~\cite{burg,thc,rocca,cohn,oumarou}. 
\paragraph{Block-encoding}
Let $A$ be a $2^m \times 2^m$ matrix. A unitary operator (circuit) $U$ acting on $n$ qubits (i.e. $U$ is a $2^n \times 2^n$ matrix) \emph{block encodes $A$} if 
\begin{equation}
	\left( \ket{\underbrace{0 \dots\dots 0}_{n - m \textrm{ zeros}} } \otimes I \right) U \left( \bra{\underbrace{0 \dots\dots 0}_{n - m \textrm{ zeros}}} \otimes I \right) = A. 
	\label{eq:block-encoding}
\end{equation}
The following Proposition summarizes elementary operations that can be performed on block-encoded matrices:
\begin{prop} \label{prop:block-encodings}
	The following operations on block encodings are possible:
	\begin{enumerate}[label=(\roman*)]
		\item If matrices $A_1, \dots A_k$ are block-encoded and $\alpha_1 \dots \alpha_k$ are arbitrary coefficients, then the linear combination $\frac{1}{| \alpha |} \sum_j \alpha_j A_j$, where $|\alpha| = \sum_j |\alpha_j|$, can also be block-encoded.
		Specifically, if $U_1 \dots U_k$ are quantum circuits acting that block encode $A_1 \dots A_k$ respectively, then it is possible to block-encode $\frac{1}{| \alpha |} \sum_j \alpha_j A_j$ with a single controlled application of each $U_j$ and $\mathcal{O}(k)$ additional gates.

		\item If a matrix $A$ is block-encoded, it is posssible to block-encode $2A^2 - I$. 
		Specifically, if $U$ is a quantum circuit acting on $n$ qubits that block encodes a $2^m \times 2^m$ matrix $A$, 
		the block-encoding of $2A^2 - I$ requires two applications of $U$ and $\mathcal{O}(n-m)$ additional gates.
		\item A quantum circuit $U$ is its own block encoding.
	\end{enumerate}
\end{prop}
For a detailed proof, we refer to section VII.A of the supplementary information in~\cite{burg}.

\paragraph{Double-factorization of the electronic Hamiltonian}
The double-factorization method begins by decomposing the two-body electronic energy tensor $g_{ijkl}$ as follows:
\begin{equation}
	g_{ijkl} = \tilde{g}_{ijkl}(R) + \epsilon_{ijkl},
	 \label{eq:g-tilde}
\end{equation}
where
\begin{equation}
 \tilde{g}_{ijkl}(R) = \sum_{r=0}^{R-1} A_{rij} A_{rkl}.
 \label{eq:first_decomp_step}
\end{equation}
Here $\tilde{g}_{ijkl}(R)$ is a decomposition of the rank-4 $g$ tensor into a sum of $R\leq N$ symmetric matrices $A_r$ (i.e. $A_{rij} = A_{rji}$). 
The truncation parameter $R$ is chosen such that $\tilde{g}_{ijkl}(R)$ approximates $g_{ijkl}$ within the error $||\epsilon||_F$, where $||.||_F$ is the Frobenius norm~\cite{oumarou}, defined as:
\begin{equation}
||\epsilon||_F	:= 
	\sqrt{\sum_{ijkl} |\epsilon_{ijkl}|^2} =
	\sqrt{ \sum_{ijkl} |g_{ijkl} - \tilde{g}_{ijkl}(R)|^2 }. 
\end{equation}	
To evaluate the quality of the approximation, the spectral norm can be used to compare the electronic Hamiltonians constructed from $g_{ijkl}$ and $ \tilde{g}_{ijkl}(R)$: $|| \sum_{ijkl} g_{ijkl} E_{ij} E_{kl} - \sum_{ijkl} \tilde{g}_{ijkl} E_{ij} E_{kl}||$. However, computing the spectral norm is as expensive as solving the electronic Schroedinger equation itself. Instead, as discussed in refs.~\cite{burg,thc,rocca}, a reasonable mitigation involves using the Frobenius norm for the electronic Hamiltonian matrix elements. Its quality has been benchmarked in refs.~\cite{rocca, thc}, showing consistency with the spectral-norm measures calculated at the CCSD(T) level of theory. For the benchmarked Hamiltonians the truncation parameter $R$ required to achieve chemical accuracy ($1$ mHartree) typically range in $R = 4N-6N$~\cite{rocca}. 

The decomposition given in eq.~\ref{eq:first_decomp_step} can be used to represent the electronic Hamiltonian in the following form:
\begin{equation} 
	\tilde{H}(R) = \sum_{i j} h_{ij} E_{ij}
	+ 	\sum_{i,j,k,l = 0}^{N - 1} \tilde{g}_{ijkl}(R) E_{ij} E_{kl} = \sum_{i j} h_{ij} E_{ij}+	\sum_{r=0}^{R-1} \left( \sum_{ij} A_{rij} E_{ij} \right)^2 =  One(A_{-1})  +
	\sum_{r=0}^{R-1} One(A_r)^2.
	\label{eq:hamiltonian_form_SF}
\end{equation}
where we adopted the notation from ref.~\cite{burg}, denoting $One(A) := \sum_{ij} A_{ij} E_{ij}$, and the one-body electronic operator is written as $One(A_{-1}) := \sum_{ij} h_{ij} E_{ij}$.
One-body operators $One(A)$ are further factorized to rank-1 operators via decomposition of the symmetric $N$-by-$N$ matrices $A$:
\begin{equation} 
	A_{ij} = \sum_t \lambda_{t} u_{ti} u_{tj}^{T},
	\quad \textrm{ where } u_{t} = (u_{ti})_{i=0 \dots N-1} \textrm{ is a unit length vector.}
	\label{eq:A_decomp}
\end{equation}
The form of eq.~\ref{eq:A_decomp} implies the quantity $\sum_t \lambda_t = \tr A$ equals the sum of eigenvalues of $A$.  In particular, the decomposition shown in eq.~\ref{eq:A_decomp} can be chosen as the eigenvalue decomposition of $A$. We note however that the double-factorization procedure does not require the vectors $u_t$ to be mutually orthogonal.

The decomposition given in eq.~\ref{eq:A_decomp} produces linear combinations of the fermionic ladder operators denoted as
\begin{equation}
	B_{u \sigma} := \sum_j u_j a_{j \sigma}. 
	\label{eq:B-op}
\end{equation}	
where $u = (u_0, \dots u_{N - 1})$ is a real unit vector and $a_{j \sigma}$ is the annihilation operator for spin-orbital $j$, defined in eq.~\ref{eq:annihilation-creation}.
The $B_{u\sigma}$ operator represents a rotation of the spin-orbital basis and satisfies fermionic commutation rules given in eq.~\ref{eq:commutation-relations-a}:
$$
B_{u \sigma}^\dg B_{u \sigma} + B_{u \sigma} B_{u \sigma}^\dg = 
\sum_i u_i^2 \left( a_{i \sigma}^\dg a_{i \sigma} + a_{i \sigma} a_{i \sigma}^\dg \right) + 
\sum_{i < j} u_i u_j \left( 
a_{i \sigma}^\dg a_{j \sigma} + a_{i \sigma} a_{j \sigma}^\dg +
a_{j \sigma}^\dg a_{i \sigma} + a_{j \sigma} a_{i \sigma}^\dg 
\right)
= I.
$$
A similar calculation shows that $B_{u \sigma}^2 = 0$. As a result we observe that: 
$$
\left( B_{u \sigma}^\dg B_{u \sigma} \right)^2 = 
B_{u \sigma}^\dg B_{u \sigma} B_{u \sigma}^\dg B_{u \sigma} =
B_{u \sigma}^\dg \left( I - B_{u \sigma}^\dg B_{u \sigma} \right) B_{u \sigma} =
B_{u \sigma}^\dg B_{u \sigma}.
$$
$B_{u \sigma}^\dg B_{u \sigma}$ is thus Hermitian and idempotent, hence represents an orthogonal projection. As a consequence, $V_{u\sigma} = 2 B_{u \sigma}^\dg B_{u \sigma} - I$ is unitary. This unitary can be efficiently implemented as a quantum circuit. Burg et al.~\cite{burg} proposed a sequence of \textit{Givens rotations} generated by the Majorana representation to decompose the unitary: $V_{u\sigma}  = i \gamma_{u \sigma 0} \gamma_{u \sigma 1}$, where 	$\gamma_{u \sigma x} := \sum_j u_j \gamma_{j \sigma x}$ and $\gamma_{j \sigma 0} = a_{j\sigma}+ a^{\dag}_{j\sigma}$,  $\gamma_{j \sigma 1} =-i(a_{j\sigma}- a^{\dag}_{j\sigma})$. Similarly, Lee et al.~\cite{thc} utilizes Givens rotations to efficiently represent $V_{u\sigma}$.  

The one-electron operators $One(A)$ can be expressed with $B_{u \sigma}$ and $B_{u \sigma}^{\dag}$ as follows:
\begin{equation}
	One(A) = 
	\sum_{ij} A_{ij} E_{ij} = 
	\sum_{tij}  \lambda_t u_{ti} u_{tj} E_{ij} =
	\sum_{t \sigma} \lambda_t 
	\left(\sum_i  u_{ti} a_{i \sigma}^\dg \right)
	\left(\sum_j  u_{tj} a_{j \sigma} \right) = 
	\sum_{t \sigma} \lambda_t B_{u_{t} \sigma}^\dg B_{u_{t} \sigma}.
	\label{eq:one-to-B}
\end{equation}
Because $2 B_{u_{t} \sigma}^\dg B_{u_{t} \sigma} - I$ is unitary for each $t=0 \dots N - 1$ and  $\sigma = 0, 1$, it represents its own block encoding. Thus, following Proposition~\ref{prop:block-encodings}(i), the block-encoding of the one-electron operators can be written as: 
\begin{equation} 
	T_1(A) := 
	\frac{1}{2 \sum_t |\lambda_{t}|}
	\sum_{t \sigma} \lambda_{t} (2 B_{u_{t} \sigma}^\dg B_{u_{t} \sigma} - I)
	= 
	\frac{1}{\sum_t |\lambda_{t}|} \left( One(A) - \tr A \right)
	= 
	\frac{1}{\Lambda} \left( One(A) - \tr A \right),
	\label{eq:temp1}
\end{equation}
where $\Lambda := \sum_t |\lambda_{t}|$.

Applying Proposition~\ref{prop:block-encodings}(ii) to eq.~\ref{eq:temp1} gives the block-encoding for the two-body operators as:
\begin{equation}
	T_2(A) :=
	2 \left( \frac{1}{\Lambda} \left( One(A) - \tr A \right) \right)^2 - I = 
	\frac{2}{\Lambda^2} \left( One(A)^2
	- 2 One(A) \tr A
	+ (\tr A)^2 \right) - I.
	\label{eqn:block-enc-one}
\end{equation}
According to eq.~\ref{eq:hamiltonian_form_SF}, there are $R+1$ matrices $A_r$ that form the decomposed electronic Hamiltonian. For each $r$, we assume a decomposition:
\begin{equation}
	A_{rij} = \sum_t \lambda_{rt} u_{rti} u_{rtj}
	\label{eq:Ar-decomp}
\end{equation}
where $\Lambda_r := \sum_t |\lambda_{rt}|$. Eigenvalue decomposition is optimal with respect to minimizing the overall block-encoding cost, as discussed in detail in Appendix~\ref{sec:appendixA}. For this reason, all rank-2-to-rank-1 tensor factorizations are assumed to be eigenvalue decompositions. 

Applying Proposition~\ref{prop:block-encodings}(i) to the operators $T_2(A_0), \dots T_2(A_{R-1})$ with weights $\Lambda_0^2 / 2, \dots \Lambda_{R - 1}^2 / 2$, gives a block-encoding for the sum of two-electron operators:
\begin{equation}
	T_3 := 
	\frac{1}{\frac{1}{2} \sum_r \Lambda_r^2} \sum_r \left( One(A_r)^2 - 2 One(A_r) \tr A_r \right)
	+ \textrm{constant term}. 
	 \label{eq:almost-tbt}
\end{equation}
\noindent
The constant term in eq.~\ref{eq:almost-tbt} is irrelevant to the block encoding and can be omitted. 
In eq.~\ref{eq:almost-tbt} the sum $\sum_r One(A_r)^2$ represents the block-encoding of the two-body electron operator given in eq.~\ref{eq:hamiltonian_form_SF}. The other component in the sum in eq.~\ref{eq:almost-tbt}, that is $- 2 \sum_r One(A_r) \tr A_r$, is incorporated into the one-body term. 
To simplify, we express the sum $\sum_r One(A_r) \tr A_r$ in terms of the orbital excitation operators defined in eq.~\ref{eq:orbital-excitation}:
$$
\sum_r One(A_r) \tr A_r = \sum_{rij} A_{rij} E_{ij} \tr A_r =
\sum_{ij} E_{ij} \cdot \left(\sum_k \sum_r A_{rij} A_{rkk} \right) = 
\sum_{ij} E_{ij} \left( \sum_k g_{ijkk} \right).
$$

\noindent
Accordingly, we define new matrix elements for the one-body operators: $h'_{ij} := h_{ij} + 2 \sum_k g_{ijkk}$ and perform eigendecomposition on $h'_{ij}$ in place of $h_{ij}$ given in eq.~\ref{eq:hamiltonian_form_SF}:
\begin{equation} 
	h'_{ij} = \sum_t \lambda_t u_{-1, i} u_{-1, j},
	\quad \textrm{ where } u_{-1} = (u_{-1, i})_{i=0 \dots N-1} \textrm{ is a unit length vector.}
	\label{eq:temp3}
\end{equation}
Using the new decomposition of the one-body terms in the electronic Hamiltonian, we repeat the block-encoding procedure shown in eq.~\ref{eq:temp1} for $h'_{ij}$ giving a block-encoding of the form $T_1(h') = \frac{1}{\Lambda_{-1}} \left( One(h') - \tr h' \right)$, 
where $\Lambda_{-1} := \sum_t |\lambda_t|$.  
Finally, we apply Proposition~\ref{prop:block-encodings}(i) to the two operators
$T_3$ and $T_1(h')$ with coefficients $\frac{1}{2} \sum_r \Lambda_r^2$ and $\Lambda_{-1}$.
This results in a block-encoding of the electronic Hamiltonian written as:
\begin{equation}
\mathbf{U}[H] = 	\frac{1}{\lambda_{DF}} 
	\left( 
	\left(\frac{1}{2} \sum_r \Lambda_r^2\right) T_3 + 
	\Lambda_{-1} T_1(h')
	\right) 
	=
	\frac{1}{\lambda_{DF}}
	\left(
	\sum_{r} One(A_r)^2 +
	\sum_{ij} h_{ij} E_{ij} 
	+ \textrm{const}
	\right)
	=
	\frac{1}{\lambda_{DF}} 
	\left( H + \textrm{const} \right),
	 \label{eq:temp2}
\end{equation}
where
\begin{equation} 
	\lambda_{DF} := \frac{1}{2} \sum_r \Lambda_r^2 + \Lambda_{-1}.
	\label{eq:lambda_DF}
\end{equation}
For the purposes of block-encoding, the constant term in eq.~\ref{eq:temp2} can be ignored.
In conclusion, the double-factorization method block encodes $H / \lambda_{DF}$, where $\lambda_{DF}$ is given in eq.~\ref{eq:lambda_DF}. 
The value of the $\lambda_{DF}$ constant given in eq.~\ref{eq:lambda_DF} depends on the decompositions shown in eq.~\ref{eq:first_decomp_step},eq.~\ref{eq:A_decomp} and eq.~\ref{eq:temp3}.

\section{Optimizing the double factorization method by norm minimization}
\label{sec:optimizing-norm}
Our objective is to minimize the value $\lambda_{DF}$, an optimization problem which can be formulated as follows:
\begin{equation} \label{problem1}
	\begin{aligned}
		\textrm{Find: } & \textrm{symmetric $N \times N$ matrices $A_0, \dots A_{R - 1}$}, \\
		\textrm{Constraint: } & ||\epsilon||_F^2 = \sum_{ijkl} | g_{ijkl} - \sum_r A_{rij} A_{rkl}|^2 \textrm{ sufficiently small}, \\
		\textrm{Minimize: } & \lambda_{DF} = \frac{1}{2} \sum_r \nuc{A_r}^2 + \nuc{h'}.
	\end{aligned}
\end{equation}
Note that $h'$ does not depend on the choice of $A_0 \dots A_{R-1}$ and therefore can be ignored in the optimization objective.
As discussed in Appendix~\ref{sec:appendixA} we use the value decomposition for each $A_r$ given in eq.~\ref{eq:Ar-decomp} and $h'$ given in eq.~\ref{eq:temp3}. Since $A_r$ and $h'$ are symmetric matrices, the appropriate block-encoding scaling constants are the respective nuclear norms:  $\Lambda_r = \nuc{A_r}$ and $\Lambda_{-1} = \nuc{h'}$. We define the nuclear norm as: $\nuc{A} := \sum_{t} \sigma_t$, where $\sigma_t$ are the singular values of $A$. In section~\ref{sec:symm-shift} we discuss a technique for lowering the value of $\lambda_{DF}$ utilizing the symmetry shift technique of Izmaylov and Loaiza~\cite{izmaylov1}.

\section{Hamiltonians parametrized by symmetry-shift operators}
\label{sec:symm-shift}
In Hamiltonian simulation, particularly for eigenvalue calculations with the quantum phase estimation algorithm~\cite{thc, burg}, the electronic Hamiltonian $H$ acts on a state $\ket{\psi}$ with a fixed \emph{number of electrons} $n_e$. The associated Hermitian operator $N_e := \sum_{i=1}^N E_{ii}$, known as \emph{number of electrons operator}, satisfies the equation: $(N_e - n_e) \ket{\psi} = 0$. 

Loaiza and Izmaylov~\cite{izmaylov1} exploit this property to replace the electronic Hamiltonian $H$ with a \textit{shifted} Hamiltonian $\tilde{H} := H + B(N_e - n_e)$, where $B$ stands for any operator. For $\tilde{H}$ to remain Hermitian, $B$ must also be Hermitian and commute with the number of electrons operator $N_e$. Non-hermiticity of $\tilde{H}$ would spoil the properties required for the double-factorization and is therefore undesirable. Following ref.~\cite{izmaylov1}, we define:
\begin{equation}
	\tilde{H} := H + \left( \sum_{ij \sigma} \xi_{ij} a_{i \sigma}^\dg a_{j \sigma} + \kappa \right) \left( N_e - n_e \right)
\end{equation}
where constants $\kappa$ and $\xi_{ij}=\xi_{ji}$ are subject to optimization. The symmetry-shifted electronic Hamiltonian $\tilde{H} =  H + B(N_e - n_e)$ is then converted to the form shown in eq.~\ref{eq:general_form_of_hamiltonian}. To preserve the $8$-fold orbital permutation symmetry of the two-electron tensor in the electronic Hamiltonian, we use the following identity:
$$ (E_{ij} + E_{ji}) (N_e - n_e) = -n_e E_{ij} -n_e E_{ji} + \sum_{r=0}^{R-1} \frac{1}{2} (E_{ij}E_{rr} + E_{ji}E_{rr} + E_{rr}E_{ij} + E_{rr}E_{ji}).
$$
Rewriting the symmetry-shifted electronic Hamiltonian in a permutationally invariant form carries a benefit that  $\tilde{H}$ takes the form of the original electronic Hamiltonian written as:
\begin{equation}
	\tilde{H} = 
	\tilde{c} + 
	\sum_{i,j = 0}^{N - 1} \tilde{h}_{ij} E_{ij} + 
	\sum_{i,j,k,l = 0}^{N - 1} \tilde{g}_{ijkl} E_{ij} E_{kl}, 
	 \label{eq:Htilde_form}
\end{equation}
where
\begin{align}
	\tilde{h}_{ij} & = h_{ij} - n_e \xi_{ij} + \kappa \delta_{ij}, \\ 
	\tilde{g}_{ijkl} & = g_{ijkl} + \frac{1}{2} \left( \xi_{ij} \delta_{kl} + \delta_{ij} \xi_{kl} \right).
\end{align}
Such a formulation enables to adopt the double-factorization technique of Burg et al.~\cite{burg} straightforwardly. For this reason, gate counting procedures, circuits and discussion of ref.~\cite{burg} remains valid in the present case. 

Using the symmetry-shifted Hamiltonian written in eq.~\ref{eq:Htilde_form}, we formulate the optimization problem outlined in section~\ref{sec:optimizing-norm} as follows:
\begin{equation} \label{problem2}
	\begin{aligned}
		\textrm{Find: }& 
		\kappa, 
		\textrm{ symmetric } N \times N \textrm{ matrix } \xi,	
		\textrm{ and symmetric } N \times N \textrm{ matrices } A_0, \dots A_{R-1}, \\
		\textrm{Constraint: } & \sum_{ijkl} | \tilde{g}_{ijkl} - \sum_r A_{rij} A_{rkl}|^2 \textrm{ sufficiently small}, 
		\textrm{ where } 
		\tilde{g}_{ijkl} = g_{ijkl} + \frac{1}{2} \left( \xi_{ij} \delta_{kl} + \delta_{ij} \xi_{kl} \right), \\
		\textrm{Minimize: } & \frac{1}{2} \sum_r \nuc{A_r}^2 + \nuc{\tilde{h}'},
		\textrm{ where } \tilde{h}_{ij} = h_{ij} - n_e \xi_{ij} + \kappa \delta_{ij}
		\textrm{ and } \tilde{h}'_{ij} := \tilde{h}_{ij} + 2 \sum_k \tilde{g}_{ijkk}.
	\end{aligned}
\end{equation}
The implementation of a solution to the problem of finding optimal $\kappa$ and $\xi_{ij}$ is discussed in 
section~\ref{sec:implementation}.

\section{Implementation details for the norm-reduction method}
\label{sec:implementation}
In this section we discuss implementation details for obtaining a solution to the optimization problem shown in eq.~\ref{problem2}. For this purpose, we define parametrized residue and block-encoding scaling constant as follows:
\begin{equation} \label{eqn:err_and_lambda}
	\begin{aligned}
		\Err (\kappa, \xi, A) & := \sum_{ijkl} \left( \tilde{g}_{ijkl} - \sum_r A_{rij} A_{rkl} \right)^2, \\
		\lambda(\kappa, \xi, A) & := \frac{1}{2} \sum_r \nuc{A_r}^2 + \nuc{\tilde{h}'}.
	\end{aligned}
\end{equation}
The problem stated in eq.~\ref{problem2} translates to the problem of minimizing $\lambda(\kappa, \xi, A)$ under the constraint $\Err(\kappa, \xi, A) = 0$. We define a cost function, similarly to the method discussed in ref.~\cite{rocca}:
\begin{equation}
	\Total(\kappa, \xi, A) := C_{\textrm{approx}} \Err(\kappa, \xi, A) + \lambda(\kappa, \xi, A),
\end{equation}
where $C_{\textrm{approx}}$ is a constant. As initial conditions for the optimization, we chosen $\kappa = 0, \xi_{ij} = 0$, and $A$ 
to be the standard double-factorization (that is, we treat the $g_{ijkl}$ tensor as a $N^2 \times N^2$ matrix, 
which is guaranteed to be symmetric and semipositive-definite, 
take eigendecomposition $g_{ijkl} = \sum_{r} d_{r} V_{rij} V_{rkl}$, then we set $A_{rij} := \sqrt{d_{r}} V_{rij}$).
The gradient descent method, as implemented in the Adam optimizer~\cite{adamoptimizer}, was used to minimize $\Total(\kappa, \xi, A)$. Our Python implementation uses the JAX package for efficient linear algebra operations \cite{jax2018github}.
Since JAX is capable of carrying out automatic differentiation, there was no need to compute the gradients symbolically. Despite the function $\Total(\kappa, \xi, A)$ being not differentiable at specific points, due to the nuclear norm involving absolute values of eigenvalues, the numerical gradient descent method performs without problems. The weight $C_\textrm{approx}$ is chosen arbitrarily and required a short trial and error procedure.

\section{Results}
\label{sec:results}
We demonstrate the performance of our implementation discussed in section~\ref{sec:implementation} on two example Hamiltonians: a Hamiltonian for the FeMoCo molecule from Reiher et al.~\cite{burg} and a Hamiltonian for the Cpd I structure in cytochrome P450~\cite{Goings2022}. We compare our present method with the standard double-factorization method of ref.~\cite{burg} (labeled XDF in Table~\ref{table}), the SCDF method of Rocca et al.~\cite{rocca}, and with the lower bounds calculated in ref.~\cite{ollitrault}. To measure the accuracy of our method, we give the value of $\Err(\kappa, \xi, A)$ in the last column. For both Hamiltonians, we find that our method achieves the values for the block-encoding scaling constant $\lambda$ lower by about $25\%$ compared to the previous works~\cite{rocca}. 

\begin{table}[!h]
	\renewcommand\arraystretch{1.3}
	\begin{tabular}{ |c|c|c|c|c|c| } 
		\hline
		Hamiltonian & 
		\hspace{1pt} N \hspace{1pt} & 
		\hspace{1pt} R \hspace{1pt}  & 
		Method 
		& $\lambda$ & 
		error \\
		\hline
		\multirow{4}*{FeMoco} &
		\multirow{4}*{54} &
		\multirow{4}*{6N} &
		XDF \cite{burg} 							& 296.0 	& $1.45 \times 10^{-5}$ \\ 
		& & & 	SCDF \cite{rocca}				& 77.9 	& -- \\ 
		& & & 	present work							& 57.9 	& $1.48 \times 10^{-5}$ \\ 
		& & & 	lowerbound \cite{ollitrault}		& 38.8	& n/a \\
		\hline
		\multirow{4}*{P450} &
		\multirow{4}*{58} &
		\multirow{4}*{6N} &
		XDF 		\cite{burg} 					& 472.9 	& $1.25 \times 10^{-6}$ \\ 
		& & & 	SCDF \cite{rocca}				& 111.0 	& -- \\ 
		& & & 	present work							& 82.8	& $2.40 \times 10^{-6}$ \\ 
		& & & 	lowerbound \cite{ollitrault}		& 49.9	& n/a \\
		\hline		
	\end{tabular}
	\caption{Comparison of the block-encoding scaling constant $\lambda$ calculated for the FeMoCo and Cpd I molecules with different methods. The columns are denoted respectively as: molecule symbol, the number of electronic orbitals $N$, the truncation parameter for the rank-4 to rank-2 factorization (cf.~\ref{eq:Ar-decomp}), the factorization method used, the value for the optimized block-encoding constant $\lambda$ and the error $\Err(\kappa, \xi, A)$.}
	\label{table}
\end{table}

\section{Conclusions}
We developed a technique for block-encoding the electronic Hamiltonian, which utilizes the Hamiltonian’s symmetries and tailored tensor-network decompositions. Our approach represents the Hamiltonian in a symmetry-shifted form introduced in ref.~\cite{izmaylov1}, preserving orbital-permutation symmetries. 
The resulting symmetry-shifted Hamiltonian is doubly-factorized with parameters optimized using a gradient descent method.
In calculations of the electronic Hamiltonian eigenvalues with quantum phase estimation, the algorithmic complexity counted in the number of T-gates is directly proportional to the block-encoding scaling constant. Our results demonstrate a 25\% reduction in this constant, compared to the best known implementations for molecules of both industrial and biological importance, including the Nitrogenase enzyme cofactor (FeMoCo) and cytochrome P450 (Cpd I structure). With the electronic orbital basis set sizes considered in this work, classical computing calculations of the electronic Hamiltonian eigenvalues for FeMoCo and P450 are presently beyond reach. Our protocol is applicable to the general class of electronic Hamiltonians and can be used to lower the fault-tolerant quantum computing cost of the Hamiltonian eigenvalue calculations.

\section{Acknowledgments}
This work is funded by the European Innovation Council accelerator grant COMFTQUA, no. 190183782.

\section{Appendix A: Optimality of eigenvalue decomposition at the second decomposition stage in the double-factorization method.}
\label{sec:appendixA}
In this appendix, we discuss factorization of rank-2 tensors (one-body operators) into rank-1 tensors. Suppose we performed the first factorization step, that is, we selected $A_{rij}$ such that~eq.\ref{eq:first_decomp_step} is satisfied. Our objective is to decompose each $A_r$ appearing in eq.~\ref{eq:Ar-decomp} and $h'$ appearing in eq.~\ref{eq:temp3} in the form given by eq.~\ref{eq:A_decomp}. In doing so, we want to minimize $\lambda_{DF}$, given by the formula in eq.~\ref{eq:lambda_DF}. Let $A$ be one of the matrices $A_0, \dots A_{R - 1}$ or $h'$. 
For each of these matrices, we solve the following optimization problem:

\begin{align*}
	\textrm{Find: } & \textrm{real numbers $\lambda_0, \dots \lambda_{N-1}$ and unit vectors $u_0 \dots u_{N-1}$}, \\
	\textrm{Constraint: } 
	& A_{ij} = \sum_{t} \lambda_t u_{ti} u_{tj}, \textrm{ or stated in another way, } A = \sum_t \lambda_t u_t u_t^T, \\
	\textrm{Minimize: } & \sum_t |\lambda_t|.
\end{align*}

The minimum is achieved by the eigenvalue decomposition, as summarized with the Proposition below:
\begin{prop} \label{thm:eigval_decomp_is_optimal}
	Let $A$ be a $N \times N$ matrix, and assume that $A = \sum_t \lambda_t u_t u_t^T$. 
	Then $\sum_t |\lambda_t| \ge \nuc{A}$, 
	where $\nuc{A}$ is the nuclear norm of $A$, 
	defined as $\nuc{A} := \sum_{t} \sigma_t$, where $\sigma_t$ are the singular values of $A$.
\end{prop}

\begin{proof}
	Let $A = U^TDV$ be the Singular Value Decomposition (SVD) decomposition of $A$. $U, V$ are orthogonal matrices, and $D$ is a diagonal matrix whose entries are the singular values $\sigma_0 \dots \sigma_{N-1}$ of $A$.
	Upon multiplying $A = \sum_t \lambda_t u_t u_t^T$ from the left by $U$ and from the right by $V^T$, we get
	\begin{equation}
		D = \sum_t \lambda_t ( U u_t) ( V u_t )^T
		\label{eq:D-matrix}
	\end{equation}
	Note that for any unit vectors $u, v$ we have $\tr uv^T = \langle u, v \rangle \in [-1, 1]$.
	We apply this property to the unit vectors $U u_t$ and $V u_t$ and take the trace over both sides in eq.~\ref{eq:D-matrix}, to get
	\begin{equation*}
		\nuc{A} = \sum_t \sigma_t = \tr D 
		= \sum_t \lambda_t \tr \left[ ( U u_t) ( V u_t )^T \right]
		\le \sum_t |\lambda_t|. \qedhere
	\end{equation*}

\end{proof}

From Proposition~\ref{thm:eigval_decomp_is_optimal} we conclude that for symmetric rank-2 tensors the nuclear norm sets a lower bound on the sum $ \sum_t |\lambda_t|$ in any decomposition of the form $A = \sum_t \lambda_t u_t u_t^T$. The nuclear norm can be calculated straightforwardly by diagonalization of the symmetric rank-2 tensor. In particular, the eigenvalue decomposition achieves this lower bound. 

\bibliography{literature.bib}
\end{document}